\documentclass[a4paper,12pt]{article}  
\pdfoutput=1
\usepackage{times}
\usepackage{amsmath,amssymb,amsthm}
\usepackage{hyperref}

\newcommand\commentout[1]{}
\def\sect#1\par{\vskip0pt plus.1\vsize\penalty-250
 \vskip0pt plus-.1\vsize\vskip7truemm
 \noindent{\large\bf#1}\par\nobreak
 \vskip3truemm{\parskip=0pt\noindent}}
\sfcode`P=1000 
\sfcode`C=1000
\def\:{\colon}

\def\endproof{\hfill\strut\nobreak\hfill\tombstone\par\medbreak}
\def\tombstone{\hbox{\lower.4pt\vbox{\hrule\hbox{\vrule
  \kern7.6pt\vrule height7.6pt}\hrule}\kern.5pt}}

\def\nats{\mathbb{N}}

\def\emptyset{\epsilon}

\newdimen \einr
\def\abs#1#2\par{{\hangafter=1\hangindent=\einr
  \noindent\hbox to\einr{\ignorespaces#1\hfill}\ignorespaces#2\par}}
\def\bye{\end{document}}

   \newtheoremstyle{theorem}
     {\topsep}
     {\topsep}
     {\rmfamily}
     {}
     {\bf}
     {.}
     {.5em}
     {}
\theoremstyle{theorem}
\newtheorem{theorem}{Theorem}
\newtheorem{definition}[theorem]{Definition}
\newtheorem{lemma}[theorem]{Lemma}

\def\seq{\sigma}   
\def\perm{{\mathcal P}}%

\newcommand{\alg}[1]{\rm\texttt{#1}}
\def\A{{\hbox{\alg{A}}}}

\def\BIT{\hbox{\alg{BIT}}}
\def\COMB{\hbox{\alg{COMB}}}
\def\TIMESTAMP{\hbox{\alg{TIMESTAMP}}}
\def\MTF{\hbox{\alg{MTF}}}
\def\LMTF{\hbox{\alg{LMTF}}}
\def\FC{\hbox{\alg{FREQUENCY COUNT}}}
\def\SPLIT{\hbox{\alg{SPLIT}}}
\def\TRANSPOSE{\hbox{\alg{TRANSPOSE}}}

\def\ON{\hbox{\alg{A}}}
\def\OFF{\hbox{\alg{OPT}}}
\def\OPT{\hbox{\alg{OPT}}}
\def\OPTBAR{\ensuremath{\overline{\alg{OPT}}}}
\def\TS{\hbox{\alg{TS}}}
\def\calA{\mathcal A}

\def\eps{\varepsilon}

\def\prob{\mathop{\mathrm{prob}}}

\newcommand{\ignore}[1]{}

\def\M{{\hat{M}}}

\title{
Optimal Lower Bounds for Projective List Update Algorithms%
\thanks{We thank two referees for helpful comments. A
preliminary version of this paper appeared in \cite{AGS00}.}
}

\author{
Christoph Amb\"uhl%
\thanks{
Dalle Molle Institute for Artificial Intelligence (IDSIA),
6928 Manno,
Switzerland.
Supported by the Swiss National Science Foundation project
200020-122110/1
``Approximation Algorithms for Machine Scheduling Through Theory and Experiments
III'' and by Hasler Foundation Grant 11099.
\newline
Email:  christoph.ambuhl@googlemail.com} \and
 Bernd G\"artner%
\thanks{Institute
of Theoretical Computer Science, ETH Z\"urich, 8092 Z\"urich,
Switzerland. \newline Email: gaertner@inf.ethz.ch}
\and Bernhard von Stengel%
\thanks{Department of Mathematics,
London School of Economics, London WC2A 2AE, United Kingdom.
Email: stengel@nash.lse.ac.uk}
}

\date{March 8, 2011}

\begin{document}
\maketitle
\begin{abstract}
The list update problem is a classical online problem, with an
optimal competitive ratio that is still open, known to be
somewhere between $1.5$ and $1.6$.
An algorithm with competitive ratio $1.6$, the smallest
known to date, is \COMB, a randomized combination of \BIT\
and the \TIMESTAMP\ algorithm \TS.
This and almost all other list update algorithms, like \MTF, are {\it
projective\/} in the sense that they can be defined by
looking only at any pair of list items at a time.
Projectivity (also known as ``list factoring'')
simplifies both the description of the
algorithm and its analysis, and so far seems to be the only
way to define a good online algorithm for lists of arbitrary
length.
In this paper we characterize all projective list update
algorithms and show that their competitive ratio is never
smaller than $1.6$ in the partial cost model.
Therefore, \COMB\ is a best possible projective algorithm in
this model.

\strut

\noindent
\textbf{Keywords:}
linear lists, online algorithms, competitive analysis.

\strut

\noindent
\textbf{AMS subject classifications:}
68W27, 
68W40, 
68P05, 
68P10. 

\end{abstract}

\newpage
\section{Introduction}

The {\it list update problem\/} is a classical online problem
in the area of self-or\-ga\-ni\-zing data structures \cite{AW98}.
Requests to items in an unsorted linear list must be served
by accessing the requested item.
We assume the {\it partial cost model\/} where accessing the
$i$th item in the list incurs a cost of $i-1$ units.
This is simpler to analyze than the original {\it full cost model\/}
\cite{RWS94} where that cost is~$i$.
The goal is to keep access costs small by rearranging the items
in the list.
After an item has been requested, it may be moved free of charge
closer to the front of the list. This is called a {\it free exchange}.
Any other exchange of two consecutive items in the list incurs
cost one and is called a {\it paid exchange}.

An {\it online\/} algorithm must serve the sequence $\sigma$
of requests one item at a time, without knowledge of future requests.
An optimum {\it offline\/} algorithm knows the entire sequence~$\sigma$
in advance and can serve it with minimum cost $\OFF(\sigma)$.
If the online algorithm serves $\sigma$ with cost $\ON(\sigma)$,
then it is called {\it $c$-competitive\/} if for a suitable
constant~$b$
\begin{equation}
\label{c-comp}
\ON(\sigma)\le c\cdot \OFF(\sigma)+ b
\end{equation}
for all request sequences $\sigma$ and all initial list states.
The infimum over all $c$ so that (\ref{c-comp}) holds for
$\ON$ is also called the {\it competitive ratio} of~$\ON$\,. If
the above inequality holds even for $b=0$, the algorithm
$\ON$ is called {\em strictly $c$-competitive} \cite{BE98}.

The {\it move-to-front\/} rule \MTF, for example, which moves each
item to the front of the list after it has been requested, is
strictly 2-competitive \cite{RWS94,ST85}. This is also the best
possible competitiveness for any deterministic online algorithm
for the list update problem~\cite{RWS94}. Another 2-competitive
deterministic algorithm is \TS, which is the simplest member of
the \TIMESTAMP\ class due to Albers \cite{A95}. \TS\ moves the
requested item $x$ in front of all items which have been requested
at most once since the last request to $x$.
\smallskip

As shown first by Irani \cite{I91},
\emph{randomized} algorithms can perform better on average.
Such an algorithm is called $c$-competitive if
\[
E[\ON(\sigma)]\le c\cdot \OFF(\sigma)+ b,
\]
for all $\seq$ and all initial list states, where the expectation
is taken over the randomized choices of the online algorithm;
this definition implies that the request sequences $\seq$ are
generated by an {\em oblivious adversary\/} that does not
observe the choices of the online algorithm.
The best randomized list update algorithm known to date is the
$1.6$-competiti\-ve algorithm \COMB\ \cite{ASW95}. It serves the
request sequence with probability $4/5$ using the algorithm \BIT\
\cite{RWS94}. With probability $1/5$, \COMB\ treats the request
sequence using \TS.

Lower bounds for the competitive ratio of randomized algorithms
are harder to find; the first nontrivial bounds are due to Karp
and Raghavan, see the remark in \cite{RWS94}. In the partial cost
model, a lower bound of~1.5 is easy to find as only two items are
needed. Teia \cite{T93} generalized this idea to prove the same
bound in the full cost model, which requires long lists. The
authors~\cite{AGS01} showed a lower bound of 1.50084
(improved to 1.50115 in \cite[p.~38]{A02}) for lists
with five items in the partial cost model, using game trees and a
modification of Teia's approach.
The optimal competitive ratio for
the list update problem (in the partial cost model) is therefore
between $1.50115$ and $1.6$, but the true value is as yet unknown.

With the exception of Irani's algorithm \SPLIT\ \cite{I91},
all the specific list update algorithms mentioned above are
\emph{projective}, meaning that the relative order of any two
items $x$ and $y$ in the list after a request sequence $\seq$ only
depends on the initial list state and the requests to $x$ and $y$
in $\seq$. The simplest example of a projective algorithm is \MTF.
In order to determine whether $x$ is in front of $y$ after $\seq$,
all that matters is whether the last request to $x$ was before
the last request to $y$. The requests to other items are
irrelevant.

A simple example of a non-projective algorithm is \TRANSPOSE,
which moves the requested item just one position further to the
front.

Projection to pairs of items, also known as ``list factoring'', is the main
tool for analyzing list update algorithms.
It has also been applied recently to other performance models of list
processing \cite {DEL09,EKL10}.

\paragraph{Our results.}
The main result of this paper (Theorem~\ref{t-irreg}) states that $1.6$ is the
best possible competitive ratio attainable by a projective algorithm.
As a tool for proving it, we develop an explicit characterization of
deterministic projective algorithms.

These results are significant in two respects. First, they show
that the successful approach of combining existing projective
algorithms to obtain improved ones has reached its limit with the
development of the \COMB\ algorithm.  New and better algorithms
(if they exist) have to be non-projective, and must derive from
new, yet to be discovered, design principles.

Second, the characterization of projective algorithms is a step
forward in understanding the structural properties of list update
algorithms. With this characterization, the largest and so far
most significant class of algorithms appears in a new, unified
way.

The complete characterization of projective algorithms turns out
to be rather involved. However, there is a simple subclass of
projective algorithms which already covers all projective
algorithms that can be expected to have a good competitive
ratio. We call them {\em critical request
algorithms\/}.
A list update algorithm is completely described by
the list state after a request sequence $\seq$ has
been served;
this can be done because we can
assume that all changes in the list state are due to paid
exchanges, as explained in further detail at the beginning
of the next section.
For critical request algorithms, the {\em unary
projections\/} to individual items suffice to describe that
list state.
For a request sequence $\seq$ and list item $x$, deleting
all requests to other items defines the unary projection
$\seq_x$, which is an $i$-fold repetition of requests
to~$x$, written as~$x^i$, for some $i\ge0$.
In Section~5 it will be necessary
to consider unary projections $x^0$ and $y^0$ of length zero
as different if the items $x$ and $y$ are different;
for the moment, this distinction does not matter.
With $L$ as the set of list items, let the set of these
unary projections be
\begin{equation}
\label{defU}
U=\{x^i\mid x\in L,~i\ge0\}.
\end{equation}

\begin{definition}[{\bf Critical request algorithm}]
\label{d-crit}
  \strut\\
A {\em deterministic} critical request algorithm is defined
by a function
\[
F:U\to \{0,1,2,\ldots\},
\quad
\hbox{ with }F(x^i)\le i
\hbox{ for any }x\in L,~ i\ge 0.
\]
We call the $F(\seq_x)$th request to $x$ in $\seq$ the
{\em critical request\/} to $x$.
If $F(\seq_x)$ is zero (for example if $\seq_x$ is the {\em empty
sequence}~$\emptyset$), then $x$ has no critical request.
In the list state after $\seq$, all items with a critical
request are grouped together in front of the items without
critical request.
The items with critical requests are ordered according to
the time of the $F(\seq_x)$th request to $x$ in $\seq$.
The later a critical request took place in the sequence, the
closer the item is to the front.
The items without critical request are placed at the end of
the list according to their order in the initial sequence.
A {\em randomized critical request algorithm\/} is a
discrete probability distribution on the set of deterministic
critical request algorithms.
\endproof
\end{definition}

As an example, consider the online algorithm for three items $a$, $b$,
and $c$ with the function $F$ shown in the following table
for requests up to four items.
\begin{center}
\begin{tabular}{c|ccccc}
$i$       & ~~0~~ &   ~~1~~  &  ~~2~~   & ~~3~~  & ~~4~~ \\
\hline
$F(a^i)$   & 0 &     1    &    0     &   2    &   2   \\
$F(b^i)$   & 0 &     0    &    2     &   2    &   4   \\
$F(c^i)$   & 0 &     1    &    2     &   2    &   2   \\
\end{tabular}
\end{center}
In the rest of this paper, list states are written as $[x_1
x_2\ldots x_n]$ where $x_1$ is the item at the front of the list.
Let the initial list state be $[abc]$.
Consider the list state after $\seq=abbcab$.
We have $F(\seq_a)=F(aa)=0$, hence $a$ does not have a
critical request.
For $b$ we have $F(\seq_b)=F(bbb)=2$, therefore the second
request to $b$ in $\seq$ is its critical request.
For $c$ we have $F(\seq_c)=F(c)=1$.
Thus after $\seq$, the list state is $[cba]$.  If we augment
$\seq$ by another request to $a$, item $a$ moves
to the front, because its critical request is the second.

Algorithms based on critical request functions are clearly
projective, since the relative order of any pair of items just
depends on the relative order of the requests to $x$ and $y$ in
$\seq$ and the relative order of $x$ and $y$ in the initial list
state.

In good online algorithms, the critical requests are very
recent, like in \MTF\ which is described by the critical
request function $F(x^i)=i$ for all items $x$.
We define the critical request \emph{relative} to the current
position by
\begin{equation}
\label{relcrit}
f(x^i)=i-F(x^i),
\end{equation}
from which the critical request function is recovered
as $F(x^i)=i-f(x^i)$.
Then \MTF\ is given by $f(x^i)=0$.
Algorithm \TS\ is described by $f(x^i)=1$ for all items
and all $i>0$ (and $f(\emptyset)=0$).
Because the \BIT\ algorithm \cite{RWS94} is randomized, its
critical requests are also randomized.
For every item $x$, its relative critical request function can
be written as $f(\emptyset)=0$ and, for $i>0$,
\begin{equation}
f(x^i)= (b_x +i )\mod 2
\end{equation}
where $b_x\in\{0,1\}$ is chosen once uniformly at random;
so for a list with $n$ items, \BIT\ is the uniform
distribution over $2^n$ different deterministic algorithms.
For \BIT, the critical request is the last or the
second-to-last request with equal probability.

The structure of the paper is as follows. In the next section, we
explain projective algorithms in more detail and how they can be
analyzed.
In Section~3, we give a characterization of so-called
$M$-regular projective algorithms, followed by the lower
bound of~$1.6$ for this important class of algorithms in
Section~4.
In Section~5, we characterize projective algorithms completely.
We extend the lower bound to the full class in Section~6.

\section{Projective Algorithms}\label{in:pa}

In order to characterize list update algorithms, we
first simplify their formal definition.
The standard definition (of the partial cost model)
considers a list state and a sequence of requests.
For each request to one of the items of the list,
the item can be accessed with access cost $i-1$ if
the item is in position~$i$, and then moved free of
charge closer to the front.
In addition, paid exchanges are allowed which can be applied
both before and after accessing the item, at a cost of one unit
for exchanging any two consecutive items.

Contrary to the claim of \cite[Theorem~3]{ST85},
paid exchanges may strictly improve costs.
For example, let the initial list state be $[abc]$ and $\seq=cbbc$.
Then an optimal algorithm moves $a$ behind $b$ and
$c$ before the first request to $c$.
This requires paid exchanges.

In order to simplify the description of a list update
algorithm, we assume that it operates using only paid
exchanges, as follows:
The list is in a certain state.
The algorithm is informed about the next request, and then
performs a number of paid exchanges, and is charged for
their cost.
It then accesses the requested item at cost $i-1$ when the
item is the $i$th item in the list, without any further
changes to the list.
This mimicks free exchanges as well:
Instead of first paying $k$ units in order to access item
$x$ and then moving it at no charge $t$ positions closer to
the front, one can first move the item $t$ positions forward
and then access the item at cost $k-t$.
In both cases, one pays exactly $k$ units.

The above description ignores paid exchanges immediately
before learning the next request; if the algorithm performs
them after the request, it has only more information.
In addition, paid exchanges also allow the transposition of
items at unrealistic low costs that are behind the requested
item in the list.
This does not matter for our lower bound considerations.

The above considerations lead to a simplified but still
equivalent model of list update algorithms:
Any deterministic online algorithm \A\ is specified by a function
\[
S^{\A}:\Sigma\to \mathcal{L}.
\]
Here, $\Sigma$ denotes the set of finite request sequences
(including the empty sequence $\emptyset$),
and $\mathcal{L}$ denotes the set of the $n!$ states the
list of $n$ items can attain.
By definition, $S^{\A}(\seq)$ is the list state after
the last request of $\seq$ has been served by algorithm~$\A$.

Consider a request sequence $\seq$ and assume it is followed
by a request to item~$x$, the extended sequence denoted by~$\seq x$.
Then the cost of serving request~$x$ is defined by:
the cost of re-arranging the list from state $S^{\A}(\seq)$
to $S^{\A}(\seq x)$ by paid exchanges, plus the cost of accessing
$x$ in state $S^{\A}(\seq x)$.

Using this notation, the initial list state can be denoted by
$S^{\A}(\emptyset)$.
We will omit the superscript \A\ in $S^{\A}(\seq)$ when
the algorithm used is determined by the context.

In order to describe projective algorithms, we have to introduce
the concept of {\em projections\/} of request sequences and list
states. Let a request sequence $\seq$ be given and fix a pair of
items $x,y$. The projection $\seq_{xy}$ of $\seq$ to $x$ and $y$
is the request sequence $\seq$ where all requests which are not to
$x$ or $y$ are removed. Similarly, $\seq_x$ is $\seq$ with all
requests other than those to~$x$ removed.

Given a list state $L$, the projection to $x$ and $y$ is obtained
by removing all items except for $x$ and $y$ from the list.
This is denoted by $L_{xy}$.

\begin{definition} \label{in:def:proj}
Let $S_{xy}(\seq)$ be the projection of $S(\seq)$ to $x$ and $y$.
A deterministic algorithm \A\ is called {\em projective}
if for all pairs of items $x,y$ and all request sequences~$\seq$
\begin{equation}\label{sxy}
S_{xy}(\seq)=S_{xy}(\seq_{xy}).
\end{equation}
A randomized algorithm is projective if all deterministic
algorithms that it chooses with positive probability are projective.

For any list update algorithm $\A$,
define the {\em projected cost\/} $\A_{xy}(\seq)$ that
$\A$ serves a request sequence~$\seq$, projected to the pair
$x,y$, as follows:
Consider all requests $z$ in $\seq$ with corresponding
prefix $\seq'z$ of $\seq$.
Then $\A_{xy}(\seq)$ is the number of times where
$S_{xy}(\seq')$ and $S_{xy}(\seq'z)$ differ (which counts
the necessary paid exchanges of $x$ and $y$; this may happen
even if $z\not\in\{x,y\}$ in case $\A$ is not projective),
plus the number of times where $z=x$ and
$S_{xy}(\seq'z)=[yx]$ or $z=y$ and $S_{xy}(\seq'z)=[xy]$.
\endproof
\end{definition}

Thus, an algorithm is projective if the relative position of
any pair of items depends only on the initial list state and the
requests to $x$ and $y$ in the request sequence.

Projective algorithms have a natural generalization, where
the relative order of any $k$-tuple of list items depends
only on the requests to these $k$ items. It turns out that for
lists with more than $k$ items, only projective algorithms satisfy
this condition. This follows from the fact that, for example
for $k=3$, $S_{xyz}(\sigma)=S_{xyz}(\sigma_{xyz})$
(so the relative position of $x$ and $y$ does not depend
on requests to~$w$), and
$S_{xyw}(\sigma)=S_{xyw}(\sigma_{xyw})$
(so the relative position of $x$ and $y$ does not depend
on requests to~$z$), which implies that
$S_{xy}(\seq)$ depends only on $\seq_{xy}$.

Already in \cite{BM85}, Bentley and McGeoch observed that \MTF\ is
projective:
Item $x$ is in front of $y$ if and only if $y$ has
not been requested yet or if the last request to $x$ took
place after the last request to $y$.

With the exception of Irani's \SPLIT\ algorithm \cite{I91},
projective algorithms are the only family of algorithms that
have been analyzed so far, typically using the following
theorem, for example in \cite{A95,ASW95,BE98}.

\begin{theorem}\label{in:thmpa}
If a (strictly) projective algorithm is
$c$-competitive on lists with two items, then it is also (strictly)
$c$-competitive on lists of arbitrary length.
\end{theorem}

\begin{proof}
Consider first an arbitrary list update algorithm $\A$.
Let $L$ be the set of list items.
Then
\begin{equation}\label{in:anyAxy}
\A(\seq)
=\sum_{\{x,y\} \subseteq L} \A_{xy}(\seq),
\end{equation}
because the costs $\A(\seq)$ are given by the update costs
for changing $S(\seq')$ to $S(\seq'z)$, which is the sum of
the costs of paid exchanges of pairs of items, plus the cost
of accessing $z$ in state $S(\seq'z)$.

For a projective algorithm $\A$ the relative behavior of a
pair of items is, according to~(\ref{sxy}), independent
of the requests to other items.
It is therefore easy to see that $\A_{xy}(\seq) = \A_{xy}(\seq_{xy})$
for projective algorithms:
Because $\A$ is projective, $\A_{xy}(\seq_{xy})$ is also the
cost of $\A$ for serving $\seq_{xy}$ on the two-item list
containing $x$ and $y$ starting from $S_{xy}(\emptyset)$.

For the algorithm $\OPT$, the term $\OPT_{xy}(\seq_{xy})$
is the cost of optimally serving $\seq_{xy}$ on the two-item
list $S_{xy}(\emptyset)$.
Hence, $\OPT_{xy}(\seq) \ge \OPT_{xy}(\seq_{xy})$.
Then
\begin{equation}\label{in:OPTxy}
\OPT(\seq)
=\sum_{\{x,y\} \subseteq L} \OPT_{xy}(\seq)
\ge \sum_{\{x,y\} \subseteq L} \OPT_{xy}(\seq_{xy})
=: \OPTBAR(\seq).
\end{equation}

Let $\A$ be a projective algorithm that is $c$-competitive on two items.
Then for every pair of items $x,y$ there is a constant
$b_{xy}$ such that for all $\seq$
\[
\A_{xy}(\seq_{xy})\le c\cdot\OPT_{xy}(\seq_{xy}) +b_{xy}.
\]
Then
\begin{eqnarray*}
\A(\seq)&=&\sum_{\{x,y\} \subseteq L} \A_{xy}(\seq_{xy})\\
&\le& \sum_{\{x,y\} \subseteq L} \left( c\cdot\OPT_{xy}(\seq_{xy}) +b_{xy} \right)\\
&\le& c\cdot\OPTBAR(\seq) + \sum_{\{x,y\} \subseteq L} b_{xy}\\
&=& c\cdot\OPTBAR(\seq) + b \\
&\le& c\cdot\OPT(\seq) + b.
\end{eqnarray*}
For the strict case, just set all $b_{xy}:=0$.
\end{proof}

Not all algorithms are projective.
Let \LMTF\ be the algorithm that moves the requested item
$x$ in front of all items which have not been requested
since the previous request to $x$, if there has been such a
request.

It is easy to prove that on lists with two items, combining
\LMTF\ and \MTF\ with equal probability would lead to a
1.5-competitive randomized algorithm.
Obviously, if \LMTF\ was projective, this bound would hold
for lists of arbitrary length.

However, \LMTF\ is not projective.
This can be seen from the request sequence $\seq=baacbc$
with initial list $L_0=[abc]$.
It holds that $S^{\LMTF}(\seq)=cab$, whereas
$S^{\LMTF}(\seq_{bc})=S^{\LMTF}(bcbc)=bca$.
Hence $ S_{bc}^{\LMTF}(\seq)\not=S_{bc}^{\LMTF}(\seq_{bc})$.

\section{Critical Requests and $M$-regular Algorithms}
\label{s-crit}

In this section, we provide technical preliminaries for our
results, and introduce $M$-regular algorithms, which move an
item to the front of the list when it has been requested $M$
times in succession.

Throughout this section, we consider deterministic projective list
update algorithms.
In order to refer to the individual requests to an item $x$,
we write unary projections as
\[
x^i=x_{(1)}x_{(2)}\ldots x_{(i)},
\]
that is, $x_{(q)}$ is the $q$th request to $x$ in $\seq$ if
$\seq_x=x^i$, for $1\le q\le i$.

Let $\perm(\seq)$ be the set of all permutations of the
sequence $\seq$.
In particular, $\perm(x^iy^j)$ consists of all sequences
with $i$ requests to $x$ and $j$ requests to~$y$.

{\em Swapping} two requests $x_{(q)}$ and $y_{(l)}$ in a
request sequence $\seq$ means that
$x_{(q)}$ and $y_{(l)}$, which are assumed to be adjacent,
change their position in $\seq$.
If two requests are not adjacent, they cannot be swapped.

\begin{definition}
\label{d-agile}
Consider a deterministic projective list update algorithm $\A$.
A pair of unary projections $x^i$, $y^j$ is called {\em agile}
 if there exist two request sequences
$\tau$ and $\tau'$ in $\perm(x^iy^j)$ with
$S_{xy}(\tau)=[xy]$ and $S_{xy}(\tau')=[yx]$.

We call a pair of requests $x_{(q)}$, $y_{(l)}$ an
{\em agile pair of $\seq$} if $x_{(q)}$ and $y_{(l)}$ are
adjacent in $\seq$ and so that $\seq'$ obtained by swapping
$x_{(q)}$ and $y_{(l)}$ in $\seq$ gives $S_{xy}(\seq')\ne S_{xy}(\seq)$.
\endproof
\end{definition}

Clearly, if $x^i$ and $y^j$ are agile, then there exists an
agile pair in at least one sequence belonging to
$\perm(x^iy^j)$.

\begin{lemma}
\label{l-adjacentinlist}
If $x_{(q)}$, $y_{(l)}$ is an agile pair of $\seq$, then $x$
and $y$ are adjacent in $S(\seq)$.
\end{lemma}

\begin{proof}
Let $\seq'$ be $\seq$ with $x_{(q)}$ and $y_{(l)}$ swapped.
Then $S_{xy}(\seq)\not=S_{xy}(\seq')$ and
$S_{st}(\seq)=S_{st}(\seq')$ for all
$\{s,t\}\subseteq L$ except $\{x,y\}$.
But this is possible only if $x$ and $y$ are adjacent in
$S(\seq)$.
\end{proof}

\begin{definition}
For every $x^i\in U$ let $R(x^i)$ be the set defined as follows:
$x_{(q)} \in R(x^i)$ if and only if there exists $y_{(l)}$
and $\seq$ with $\sigma_x=x^i$
such that $x_{(q)},y_{(l)}$ is an agile pair of~$\seq$.
\endproof
\end{definition}

\begin{lemma}
\label{l-unique}
Let $x^i$ be unary projection and suppose that $x^i$ forms
agile pairs involving at least two distinct items.
Then $|R(x^i)|=1$.
\end{lemma}

\begin{proof}
Obviously, $|R(x^i)|>0$.
Suppose that $|R(x^i)|>1$;
we will show that this leads to a contradiction.
Then there are two distinct items $y,z$ and
a sequence $\tau\in \perm(x^iy^j)$ with an
agile pair $x_{(q)},y_{(l)}$, and similarly
$\lambda\in \perm(x^iz^k)$ with an agile pair
$x_{(q')},z_{(m)}$ with $q\not=q'$
(if $q=q'$ for all choices
of $y_{(l)}$ and $z_{(m)}$, then $|R(x^i)|=1$).
We insert $k$ requests to $z$ into $\tau$, but not between
$x_{(q)}$ and $y_{(l)}$, to create a sequence $\seq$ with
$\seq_{xy}=\tau$ and $\seq_{xz}=\lambda$ in which both
$x_{(q)},y_{(l)}$ and $x_{(q')},z_{(m)}$ are adjacent pairs.

Swap the agile pair $x_{(q)},y_{(l)}$ in $\seq$ to obtain
$\seq'$ with $\{S_{xy}(\seq),S_{xy}(\seq')\}=\{[xy],[yx]\}$.
We have $\seq_{xz}=\seq'_{xz}=\lambda$.
Suppose that $S_{xz}(\seq)=[zx]$ and $S_{xy}(\seq)=[xy]$
(and hence $S_{xyz}=[zxy]$), or that $S_{xz}(\seq)=[xz]$ and
$S_{xy}(\seq)=[yx]$ (and hence $S_{xyz}=[yxz]$), otherwise
exchange $\sigma$ and~$\sigma'$.
Now consider the sequence $\sigma''$ obtained from $\seq$ by
swapping both agile pairs $x_{(q)},y_{(l)}$ and
$x_{(q')},z_{(m)}$.
This reverses the three-element list $S_{xyz}$, that is,
$\{S_{xyz}(\seq),S_{xyz}(\seq'')\}=\{[zxy],[yxz]\}$,
so that $S_{yz}(\seq)\ne S_{yz}(\seq'')$, but
$\seq_{yz}=\seq''_{yz}$, which contradicts the projectivity
of the algorithm with respect to $y$ and~$z$.
\end{proof}

\begin{lemma}\label{l-beforeafter}
If $x_{(q)},y_{(l)}$ is an agile pair in $\lambda\in\perm(x^iy^j)$
and $|R(x^i)|=1$ and $|R(y^j)|=1$, then
the only swap of requests that can change the relative order
of $x$ and $y$ in a request sequence in $\perm(x^iy^j)$
is swapping $x_{(q)}$ and $y_{(l)}$, and this changes
$S_{xy}(\seq)$ in any such sequence $\seq$ where
$x_{(q)}$ and~$y_{(l)}$ are adjacent.
\end{lemma}

\begin{proof}
Only $x_{(q)}$ and $y_{(l)}$ can be swapped to affect the
order of $x$ and $y$ because $|R(x^i)|=|R(y^j)|=1$.
If the lemma does not hold, then there exists a sequence
$\seq$ in $\perm(x^iy^j)$ in which we can swap $x_{(q)}$ and
$y_{(l)}$ to obtain $\seq'$ with
$S_{xy}(\seq)=S_{xy}(\seq')$.
Then we can obtain any sequence in $\perm(\seq)$
by successively transposing adjacent requests,
starting from either $\seq$ or $\seq'$, without ever swapping
$x_{(q)}$ and $y_{(l)}$.
Thus, the relative order of $x$ and $y$ would be the same
for all request sequences in $\perm(x^iy^j)$.
But we know that swapping $x_{(q)}$ and $y_{(l)}$ changes
$S_{xy}(\lambda)$.
This is a contradiction.
\end{proof}

In this and the next section, we consider online list update
algorithms that move an item to the front of the list
after sufficiently many consecutive requests to that item.
This behavior is certainly expected for algorithms with a
small competitive ratio.
In this section, we show that such algorithms, which we call
$M$-regular, can be characterized in terms of ``critical
requests''.
In the next section, we use this characterization to show
that such algorithms are at best $1.6$-competitive.

\begin{definition}
\label{def:regular1}
For a given integer $M>0$, a deterministic algorithm is
called \emph{$M$-regular} if for each item $x$ and each
request sequence $\seq$, item $x$ is in front of all
other items after the sequence $\seq x^M$.

A randomized algorithm is called $M$-regular if it is a
discrete probability distribution over deterministic $M$-regular
algorithms.
\endproof
\end{definition}

The algorithms discussed at the end of the introduction are
all 1-regular or 2-regular.
A projective algorithm that is not $M$-regular is \FC,
which maintains the items sorted according to decreasing
number of past requests; two items which have been requested
equally often are ordered by recency of their last request,
like in \MTF.
Hence, after serving the request sequence $x^{M+1}y^M$, item
$x$ is still in front of $y$, which shows that \FC\ is not
$M$-regular for any~$M$.
Projective algorithms that are not $M$-regular are characterized in
Section~5 below, but such ``irregular'' behavior must
vanish in the long run for any algorithm with a good
competitive ratio (see Section~6).
Hence, the important projective algorithms are $M$-regular.

The following theorem asserts the existence of critical
requests, essentially the unique element of $R(x^i)$ in
Lemma~\ref{l-unique}, for those unary projections $x^i$
where this lemma applies.
For projectivity, the list items may also be maintained in
reverse order, described as case (b) in the following
theorem; competitive algorithms do not behave like this, as
we will show later.

\begin{theorem}
\label{thm:operation}
Let $\A$ be a deterministic projective algorithm
over a set $L$ of list items.
Then there exists a function
\begin{equation*}
F : U \to \nats,
\qquad F(x^i)\leq i
\qquad\hbox{for all }i
\end{equation*}
so that the following holds:

Let $Q$ be a set of unary projections containing
projections to at least three different items.
Let all unary projections to different items in $Q$ be
pairwise agile.
Then one of the following two cases (a) or (b) applies.

\begin{itemize}
\item[(a)] For all pairs of unary projections $x^i,y^j$ from $Q$ it holds that if $q=F(x^i)$ and $l=F(y^j)$, then
\begin{equation}\label{e-casea}
S_{xy}(\seq)=
\begin{cases}
  [xy]& \text{if } x_{(q)} \text{ is requested after } y_{(l)} \text{ in } \seq\\
  [yx] & \text{if } x_{(q)} \text{ is requested before } y_{(l)} \text{ in } \seq
\end{cases}
\end{equation}
\item[(b)] For all pairs of unary projections $x^i,y^j$ from $Q$ it holds that if $q=F(x^i)$ and $l=F(y^j)$, then
\begin{equation}\label{e-caseb}
S_{xy}(\seq)=
\begin{cases}
  [xy]& \text{if } x_{(q)} \text{ is requested before } y_{(l)} \text{ in } \seq\\
  [yx] & \text{if } x_{(q)} \text{ is requested after } y_{(l)} \text{ in } \seq
\end{cases}
\end{equation}
\end{itemize}
\end{theorem}

\begin{proof}
Since all pairs of unary projections in $Q$ are pairwise
agile, we can conclude $|R(x^i)|=1$ for all $x^i\in Q$ by
Lemma~\ref{l-unique}.
This allows us to define $F(x^i)=q$ if $x_{(q)}\in R(x^i)$.
 From Lemma~\ref{l-beforeafter} we can conclude that for
every pair $x^i,y^j$, either~(\ref{e-casea})
or~(\ref{e-caseb}) holds.

It remains to prove that
either all pairs are operated by (\ref{e-casea}) or by~(\ref{e-caseb}).
If this was not the case, then it is not hard to see that
one can construct a sequence $\seq$ which has a pair of
critical requests adjacent to each other in $\seq$ (which
define an agile pair) without
the corresponding items being adjacent in $S(\seq)$, which
contradicts Lemma~\ref{l-adjacentinlist}:
For example, suppose
$F(x^i)=q$, $F(y^j)=r$, and $F(z^k)=s$, consider $\seq$ in
$\perm(x^iy^jz^k)$ so that $\seq$ has the three consecutive
requests $x_{(q)}z_{(s)}y_{(r)}$, and assume that
$S_{xyz}(\seq)=[xyz]$ because $x_{(q)}$ is requested before
$y_{(r)}$ according to (\ref{e-caseb}) and because $y_{(r)}$
is requested after $z_{(s)}$ according to (\ref{e-casea});
then the critical requests $x_{(q)}$ and $z_{(s)}$ are
adjacent in $\seq$ but $x$ and $z$ are not adjacent in
$S(\seq)$.
\end{proof}

The following theorem asserts that, in a list with at least
three items, an $M$-regular algorithm operates according to
critical requests as in Definition~\ref{d-crit} for all
pairs of items that have been requested $M$ or more times.
That is, case (b) of Theorem~\ref{thm:operation}, where the
list items are arranged backwards, does not apply.
In addition, the critical request to any item must be one
of the last $M$ requests to that item, which means that the
relative critical request $f(x^i)$ in (\ref{relcrit}) is less
than~$M$.

\begin{theorem}
\label{thm:operationMreg}
Let $\A$ be a deterministic projective $M$-regular algorithm
over a set $L$ of at least three list items.
Then there exists a function
\begin{equation*}
F : U \to \nats,
\qquad F(x^i)\leq i
\qquad\hbox{for all }i
\end{equation*}
so that the following holds. Let $x,y\in L$. Let $\seq$ be
any request sequence with $|\seq_x|\ge M$ and
$|\seq_y|\ge M$.
Then, with $q=F(x^i)$ and $l=F(y^j)$,
\[
S_{xy}(\seq)=
\begin{cases}
  [xy]& \text{if } x_{(q)} \text{ is requested after } y_{(l)} \text{ in } \seq\\
  [yx] & \text{if } x_{(q)} \text{ is requested before } y_{(l)} \text{ in } \seq
\end{cases}
\]
Moreover, with $f(x^i)$ defined as in $(\ref{relcrit})$, we
have $f(x^i)<M$ for all~$i$.
\end{theorem}

\begin{proof}
Let $Q$ be the set of all unary projections $x^i$ with $i\ge M$.
This set has all the properties of the set $Q$ in
Theorem~\ref{thm:operation}, where clearly case (a) applies
because $S_{xy}(x^My^M)=[yx]$.
Because $\A$ is $M$-regular, for $i\ge M$ the critical
request $F(x^i)$ is one of the last $M$ requests to~$x$,
which shows that $f(x^i)<M$; for $i<M$ this holds trivially.
\end{proof}

\section{The Lower Bound for $M$-regular Algorithms}\label{sec:lowerbound}
In this section, we use Theorem~\ref{thm:operation} to prove
the following result.

\begin{theorem}
\label{t:Mreg}
No $M$-regular projective algorithm is better than
$1.6$-competitive.
\end{theorem}

We first give an outline of the proof. Given any $\eps>0$
and $b$, we will show that there is a discrete probability
distribution $\pi$ on a finite set $\Lambda$ of request
sequences so that
\begin{equation}\label{e:yao}
  \sum_{\lambda \in \Lambda} \pi(\lambda)
  \frac{\A(\lambda)}{\OFF(\lambda)+b}\ge 1.6-\eps,
\end{equation}
for any deterministic $M$-regular algorithm $\A$. Then
\emph{Yao's theorem} \cite{Y77} asserts that also any
randomized $M$-regular algorithm has competitive ratio
$1.6-\eps$ or larger. This holds for any $\eps>0$, so the
competitive ratio is at least~$1.6$. This ratio is achieved by
\COMB, and therefore $1.6$ is a tight bound for the
competitive ratio of $M$-regular algorithms.

All $\lambda \in \Lambda$ will consist only of requests to
two items $x$ and~$y$. In what follows, let $\M\ge M$ and
$\M\ge 3$ and let the request sequence $\phi$ be
\def\g#1{{#1}\,}
\begin{equation}
\label{e:defphi}
\phi:=
\g{x^\M}
\g{yx^\M}
\g{y^\M}
\g{xy^\M}
\g{x^\M}
\g{yxyx^\M}
\g{y^\M}
\g{xyxy^\M}.
\end{equation}
By the last observation in Theorem~\ref{thm:operationMreg},
$x$ will be in front of the list after any subsequence
$x^\M$ of requests, and $y$ after any subsequence~$y^\M$.
The purpose of the following construction is to obscure to
the algorithm (which operates according to critical requests
defined by the unary projections) the exact location of a
request to $x$ or $y$ in a repetition of~$\phi$.

Let $K$ and $T$ be positive integers
and let $H$ be the number of requests to~$x$ (and to~$y$)
in~$\phi$, that is,
\begin{equation}
\label{defH}
  H:=|\phi|/2=4\M+4.
\end{equation}
Then the set $\Lambda$ of sequences in (\ref{e:yao}) is given by
\begin{equation}
\label{e:uniformdist}
\Lambda = \Lambda(K,T):=\{
x^{\M+t} y^{\M+h}\phi^K \mid  0\leq h<H, 0\leq t< HT\},
\end{equation}
where $\pi$ chooses any $\lambda$ in $\Lambda$ with equal
probability $\pi(\lambda)=1/H^2T$.
Note that in (\ref{e:uniformdist}), $K$ is the number of
repetitions of~$\phi$, the number $H$ depends on $\M$ but is
otherwise constant, $h$ creates a prefix for $y$ so as to
achieve any possible position inside $\phi$ for a given
request to~$y$, and $T$ is a second parameter that defines
the range of~$t$ so that the number of requests to $x$ can
vary widely relative to~$y$; it is not necessary to
introduce such a parameter for~$y$.

It is easy to see that $\OFF$ pays exactly ten units for
each repetition of $\phi$ (which always starts in offline
list state $[yx]$).
Assuming that the initial list state is also $[yx]$, all
sequences in $\Lambda$ have offline cost $10K+2$.
This and the fact that $\pi(\lambda)$ for $\lambda\in\Lambda$
is constant allows us to show (\ref{e:yao}) once we can prove
-- which we will do in the course of our argument --
\begin{equation}
\sum_{\lambda \in \Lambda}\A(\lambda)
\geq 16KH^2T - o(KH^2T),
\label{e:yao2}
\end{equation}
because then
\begin{eqnarray*}
\sum_{\lambda \in \Lambda} \pi(\lambda) \frac{\A(\lambda)}{\OFF(\lambda)+b}
&=& \frac{\sum_{\lambda \in \Lambda}\A(\lambda)}
{\sum_{\lambda \in \Lambda}(\OFF(\lambda)+b)}
\geq \frac{16KH^2T - o(KH^2T)}{(10K+2+b)H^2T}
\geq 1.6 - \eps
\end{eqnarray*}
for $K$ and $T$ large enough.

Recall that by Theorem~\ref{thm:operationMreg}, the
algorithm uses critical requests that depend only on the
unary projections $x^i$ and $y^j$ to $x$ and $y$ of a
sequence in $\Lambda$.
We refer to the pair $(i,j)$ as a \emph{state}, according to
the following definition.

\begin{definition}
A request sequence $\seq$ ends at \emph{state}
$(i,j)$ if $|\seq_x|=i$ and $|\seq_y|=j$.
The request sequence $\lambda$ \emph{passes} state $(i,j)$
if there is a proper prefix $\seq$ of $\lambda$, with
$\lambda=\seq\tau$ for non-empty $\tau$, so that $\seq$ ends at $(i,j)$.
The request in $\lambda$ \emph{after} $(i,j)$ is
the first request in~$\tau$.
\endproof
\end{definition}

\begin{definition}
\label{d-pass}
Let $\A_{\lambda}(i,j)$ denote the online cost
of serving the requests in $\lambda$ after $(i,j)$.
If $\lambda$ does not pass $(i,j)$, let $\A_{\lambda}(i,j)=0$.
\endproof
\end{definition}

We will show that the set $\Lambda$ in (\ref{e:uniformdist})
is constructed in such a way that almost all states which
are passed by some sequence $\lambda$ in $\Lambda$ are
so-called \emph{good states}, defined as follows.

\begin{definition}
\label{d-good}
A state $(i,j)$ is called \emph{good} if for every proper
prefix $\seq$ of~$\phi$ (that is, $0\le|\sigma|<2H$) there
exist unique $h,k,t$ with $0\le h<H$, $0\le k<K$ and
$0 \le t<HT$ so that $x^{\M+t}y^{\M+h}\phi^k\seq$ ends at
state~$(i,j)$.
\endproof
\end{definition}

Note that $H$ is the number of requests to $y$ in $\phi$,
so given $(i,j)$ and the prefix $\sigma$ of $\phi$ in
Definition~\ref{d-good}, there is at most one choice of $h$
and $k$, and therefore at most one $t$, so that
$x^{\M+t}y^{\M+h}\phi^k\seq$ ends at state $(i,j)$ (see also
(\ref{good}) below).
The state $(i,j)$ is good if these $h,k,t$ exist for all
proper prefixes $\sigma$ of $\phi$, which means that each
position inside the repetition of $\phi$ in the sequence
chosen randomly from $\Lambda$ is equally likely.

The following Lemma~\ref{l:good16} states that good states
incur large costs.
After that we prove that almost all states are good and thus
complete the proof of Theorem~\ref{t:Mreg}.

\begin{lemma}
\label{l:good16}
Let $(i,j)$ be a good state. Then
\[
\sum_{\lambda\in\Lambda}\A_\lambda(i,j) \ge 16.
\]
\end{lemma}

\begin{proof}
Consider any sequence $\lambda$ in $\Lambda$ so that $\lambda$ passes
$(i,j)$; there are $2H$ such sequences by
Definition~\ref{d-good}.
The request in $\lambda$ after $(i,j)$ is some request in~$\phi$.
The cost $\A_\lambda(i,j)$
of serving that request depends on whether the requested
item $x$ or~$y$ is in front or not.
This, in turn, is determined by the terms $f(x^i)$ and $f(y^j)$
as defined in~(\ref{relcrit}),
which determine the relative critical requests to $x$ and $y$
in~$\lambda$.
Recall that the item with the more recent critical request
is in front, and that $f(x^i)$ and $f(y^j)$ are less
than~$\M$ by Theorem~\ref{thm:operationMreg}.

Because $(i,j)$ is a good state, we obtain exactly all the
requests in $\phi$ as the requests after $(i,j)$ in
$\lambda$ when considering all $\lambda$ in $\Lambda$ that
pass~$(i,j)$.
Therefore, the total cost
$\sum_{\lambda\in\Lambda}\A_\lambda(i,j)$ is the cost of
serving exactly the requests in $\phi$ according to
the critical requests as given by $f(x^i)$ and $f(y^j)$.

\newcommand{\ontop}[2]{\genfrac{}{}{0pt}{}{#1}{#2}}
\begin{table}[ht]
\arraycolsep2.5pt
\noindent\strut\hfill
$
\begin{array}{r|r@{\,\,}||l|l|l|l|l|l|l|l||r}
f(x^i) &
f(y^j) &
 x^{\M} &
yx^{\M} &
y^{\M} &
xy^{\M} &
x^{\M} &
yxyx^{\M} &
y^{\M} &
xyxy^{\M} &
\sum_{\lambda\in\Lambda}\A_\lambda(i,j)
\\
\hline
0 & 0 &
1..   &
11..  &
1..  &
11..  &
1..  &
1111..  &
1..  &
1111..  &
16 \qquad \\

0 & {\ge}1 &
1..  &
1..  &
11..  &
111..  &
1..  &
101..  &
11..  &
11011..  &
{\ge}16 \qquad \\

1 & 1 &
11..  &
1..  &
11..  &
1..  &
11..  &
1011..  &
11..  &
1011..  &
16 \qquad \\

1 & {\ge}2 &
11..  &
1..  &
111..  &
1..  &
11..  &
101..  &
111..  &
10111..  &
{\ge}18 \qquad \\

{\ge}2 & {\ge}2 &
111..  &
1..  &
111..  &
1..  &
111..  &
101..  &
111..  &
101..  &
{\ge}18 \qquad \\
\end{array}
$
\hfill\strut

\caption{Online costs $\A_{\lambda}(i,j)$ for all $\lambda$
that pass a good state $(i,j)$, which are the costs of
serving the requests in~$\phi$.
They depend on the relative critical requests $f(x^i)$ and $f(y^j)$.
}
\label{phitable}
\end{table}

The rows in Table~\ref{phitable} show the costs
$\A_{\lambda}(i,j)$ for the possible combinations of
$f(x^i)$ and $f(y^j)$, up to symmetry in $x$ and~$y$
(explained further at the end of this proof).
For example, consider the first case
$f(x^i)=0$ and $f(y^j)=0$,
where the critical request to an item is always the most
recent request to that item, which is the \MTF{} algorithm.
Suppose that the request after $(i,j)$ is the first request,
to~$x$, in the subsequence $xy^{\M}$ of~$\phi$.
The critical request to~$x$ is the last request to $x$
earlier in $yx^{\M}$,
and the critical request to~$y$ is the last request to $y$
earlier (and more recent) in $y^{\M}$.
The critical request to $y$ is later than that to~$x$,
so $y$ is in front of $x$, and serving $x$ incurs cost~$1$,
which is the first $1$ in the table entry $11..$ in the
column for $xy^{\M}$.
The second $1$ in $11..$ is the cost of serving the first~$y$.
It is $1$ because here the critical request to~$x$ is more
recent than the critical request to~$y$.
The ``$..$'' in $11..$ correspond to the costs of later
requests to $y$ in $y^{\M}$, which are zero for $f(x^i)=0$
and $f(y^j)=0$ (so for $\M=4$ the complete cost sequence
would be $11000$).
In a good state, each cost $0$ or $1$ in the table (in
correspondence to the respective position in~$\phi$) is
incurred by a sequence~$\lambda$ in $\Lambda$.

By construction of $\Lambda$, the requests before
$x^{\M}$ in the first column of
Table~\ref{phitable} are of the form $y^{\M}$,
so $y$ is in front of~$x$, and the first request of~$x^\M$
has always cost~$1$.

In the second row in Table~\ref{phitable}, $f(x^i)=0$ and
$f(y^j)\ge1$; if $f(y^j)=1$, then the request to $y$ is
handled as in the \TS{} algorithm.
As an illustration of a more complicated case,
consider the subsequence $xyxy^{\M}$ of $\phi$ in the
last column, with associated costs $11011..$.
The first~$1$ is the cost of serving the first request
to~$x$, because the preceding requests are $\M\ge M$ requests
to~$y$ in $y^{\M}$ and because the algorithm is $M$-regular,
which means $f(y^j)<M$ by Theorem~\ref{thm:operationMreg},
so $y$ is in front of~$x$.
Because $f(x^i)=0$, the cost of serving the first $y$ in
$xyxy^{\M}$ is also~$1$, because $x$ is in front of~$y$.
The second request to~$x$ has cost~$0$ (the first $0$ in
$11011..$) because $y$ is not moved in front of~$x$ (the
critical request to~$y$ is earlier than that to~$x$ because
$f(y^j)\ge1$).
The next two costs $11$ are for the second and third request
to~$y$ in $xyxy^{\M}$, because the critical request to~$x$
is more recent.

The rows in Table~\ref{phitable} describe all cases for
$f(x^i)$ and $f(y^j)$ with $i\le j$.
They describe in fact all possible cases because for each
column in Table~\ref{phitable} there is another column with
$x$ and $y$ interchanged, where the costs for requests to
$x$ and $y$ apply in the same manner when $x$ is exchanged
with~$y$.
The respective costs in Table~\ref{phitable} are easily verified.
The right column shows that the total cost
$\sum_{\lambda\in\Lambda}\A_\lambda(i,j)$ is at least 16 in
all these cases, which proves the claim.
\end{proof}

The preceding proof of Lemma~\ref{l:good16} also shows that
1.6-competitive algorithms can only be expected when the
relative critical requests fulfill $f(x^i)\in \{0,1\}$, as
in the \MTF\ and \TS\ algorithms.

\noindent
\begin{proof}[Proof of Theorem~\ref{t:Mreg}.]
We only have to prove~(\ref{e:yao2}), which
we will do by showing
\begin{equation}
\label{yao3}
\sum_{\lambda \in \Lambda}\A(\lambda)\ge
\sum_{(i,j)\textrm{ good}}\sum_{\lambda\in\Lambda}\A_\lambda(i,j)
\ge 16KH^2T - o(KH^2T).
\end{equation}
The first inequality in (\ref{yao3}) is immediate.
For the second inequality we use Lemma~\ref{l:good16}.
It suffices to show that the number of good states is at least
\begin{equation*}
KH^2T - o(KH^2T).
\end{equation*}
By Definition~\ref{d-good}, state $(i,j)$ is good if and only if
\begin{equation}
\label{good}
\begin{array}{rcl}
i &=& \M+ t     +  kH + |\seq_x|, \\
j &=& \M+ h     +  kH + |\seq_y|,
\end{array}
\end{equation}
or equivalently
\begin{equation}
\label{hkt}
\begin{array}{rcl}
t &=&  i + h - j - (|\seq_x|-|\seq_y|), \\
h+kH &=& j - \M - |\seq_y|.
\end{array}
\end{equation}
For $0\le k<K$ and $0\le h< H$, the term $h+kH$ takes
the values $0,\ldots, KH-1$.
The second equation in (\ref{hkt}) therefore has a unique
solution in $h,k$, for any $\seq$ (where $0\le |\seq_y|<H$)
whenever $\M+H-1\le j <\M+KH$.
Because by (\ref{e:defphi}), $0\le |\seq_x|-|\seq_y|<H$,
the first equation in (\ref{hkt}) has a unique solution $t$
in $\{0,\ldots, HT-1\}$ if $j+H-1\le i \le j+HT-H$,
for every fixed~$j$.
Hence the number of good states is at least
\[
(KH-H+1)\cdot (HT-2H+2)= KH^2T-o(KH^2T)
\]
because for sufficiently large $K$ (the number of
repetitions of $\phi$) and $T$ (the number of initial
repetitions of $x$) all other terms are arbitrarily small
relative to $KH^2T$.
\end{proof}

\section{The Full Characterization}

In this section, we give the full characterization of
deterministic projective algorithms.
We consider the set $U$ of unary projections of request
sequences defined in~(\ref{defU}) as the set of nodes of
the directed graph $G=(U,E)$ with arcs $(x^i,y^j)$ in
$E$ whenever there is a request sequence $\seq$ in
$\perm(x^iy^j)$ with $S(\seq)=[xy]$.

For any two distinct items $x$ and $y$ and any $i,j\ge0$,
there is at least one arc between $x^i$ and~$y^j$.
If the pair $x^i,y^j$ is agile according to
Definition~\ref{d-agile}, then there are arcs in both
directions.
Only pairs of nodes of the form $x^i$, $x^j$ do not have
arcs between them.

Let $\cal W$ be the set of strongly connected components of~$G$,
and let $C(x^i)$ be the strongly connected component that $x^i$
belongs~to.
We think of $C(x^i)$ as a ``container'' that contains $x^i$
and all other unary projections $y^j$ with $C(y^j)=C(x^i)$.

There exists a total order $<$ on these
containers so that $C(x^i)<C(y^j)$ if $S_{xy}(\seq)=[xy]$
after serving any $\seq\in\perm(x^iy^j)$.
To see this, we define the following binary relation $P$ on
$\cal W$: Let $C(x^i)\,P\,C(y^j)$ if there is a path in $G$
from $x^i$ to $y^j$.
Then $P$ defines a partial order on $\cal W$.  It
is acyclic because cycles in $G$ belong to strongly connected
components, which are the elements of $\cal W$.
The only pairs of containers which are not ordered in $P$
are those of the form $\{x^i\}$, $\{x^j\}$ for which there
does not exist a container $C(y^k)$ with
$C(x^i)<C(y^k)<C(x^j)$ or $C(x^j)<C(y^k)<C(x^i)$.
By stipulating $\{x^i\}<\{x^j\}$ if and only if $i<j$ for
such pairs, we can extend $P$ to the desired total order
$<$.

A specific case is given by the empty unary projections
$x^0$ for items~$x$:
Note that $x^0$ and $y^j$ for any $j\ge 0$ are never in the
same container because $\perm(x^0y^j)$ contains only a
single sequence $\seq=y^j$; the state $S_{xy}(\seq)$ is
therefore either $[xy]$ or $[yx]$, so there cannot be paths
in both directions between $x^0$ and~$y^j$ in~$G$.
Hence $C(x^0)=\{x^0\}$, and $C(x^0)<C(y^0)$ if and only if $x$
is in front of $y$ in the initial list.

In summary, for a request sequence $\seq$, the total
order $<$ on $\cal W$ determines the list order between two
items $x$ and $y$ whose unary projections $\seq_x$ and
$\seq_y$ belong to different containers in~$\cal W$.

If $\seq_x$ and $\seq_y$ belong to the same container, then
the list order between $x$ and $y$ can be described by
essentially two possibilities.
First, if the container contains only projections to at most
two items $x$ and~$y$, nothing further can be said because
the relative order between $x$ and $y$ for these requests
is arbitrary without violating projectivity (for the same
reason that on a two-item list, any algorithm is projective);
the set of these containers will be denoted by~${\cal W}_2$.

Second, if a container contains unary projections for three or more
distinct items, then the algorithm's behavior can be described by
critical requests similar to Theorem~\ref{thm:operation};
the set of such containers will be denoted by $\cal W^+$.
There is a symmetric set $\cal W^-$ where the algorithm
behaves in the same manner but with the list order reversed
(which does not define competitive algorithms).

These assertions are summarized in the following theorem.

\begin{theorem}
\label{t-containers}
Consider a deterministic projective list update algorithm.
Then there are pairwise disjoint sets
${\cal W}^+$,
${\cal W}^-$,
${\cal W}_2$
whose union is $\cal W$
and a total order $<$
on $\cal W$ and a function $C:U\to \cal W$ with

\begin{itemize}
\item[(I)] $C(x^0)=\{x^0\}\in {\cal W}_2$ for all $x\in L$;
\item[(II)] for any three items $x,y,z$, if
$C(x^i)=C(y^j)=C(z^k)=w$, then $w\not\in{\cal W}_2$.
\end{itemize}

\noindent Furthermore, if $C(x^i)\not\in {\cal W}_2$, then there exists
$F(x^i)\in \{1,\ldots,i\}$ with the following properties:
For all request sequences $\seq$
with $\seq_{x}=x^i$ and $\seq_{y}=y^j$,

\begin{itemize}
\item[(III)]
if $C(x^i)<C(y^j)$ then $S_{xy}(\seq)=[xy]$;
\item[(IVa)]
if $C(x^i)=C(y^j)\in \cal W^+$ then $S_{xy}(\seq)=[xy]$
if and only if the $F(x^i)$th request to $x$ is {\em after\/} the
$F(y^j)$th request to $y$ in $\seq$;
\item[(IVb)]
if $C(x^i)=C(y^j)\in \cal W^-$ then $S_{xy}(\seq)=[xy]$
if and only if the $F(x^i)$th request to $x$ is {\em before\/} the
$F(y^j)$th request to $y$ in $\seq$.
\end{itemize}
\end{theorem}

\begin{proof}
The set $\cal W$ and the order $<$ have been defined above
with the help of the graph $G$, which shows (III).
We have also shown (I) above.

As before, let ${\cal W}_2$ be the set of containers with
unary projections to at most two distinct items, which
implies~(II).

It remains to show (IVa) and (IVb).
Consider a request sequence $\seq$
with $\seq_{x}=x^i$ and $\seq_{y}=y^j$.
Let $C(x^i)=C(y^j)\not\in {\cal W}_2$, so that there is a
third item $z\not\in\{x,y\}$ with $C(x^i)=C(y^j)=C(z^k)$.
We want to apply Lemma~\ref{l-unique}.
To this end, we first show the ``mixed transitivity''
(note that $x,y,z$ are distinct items)
\begin{equation}
\label{transitive}
(x^i,y^j)\in E \quad\text{and}\quad (y^j,z^k)\in E
\quad\Longrightarrow\quad(x^i,z^k)\in E.
\end{equation}
Let $(x^i,y^j)\in E$, so that $S_{xy}(\seq)=[xy]$
for some $\seq\in \perm(x^iy^j)$.
If $(y^j,z^k)\in E$, then one can insert $k$ requests to $z$
into $\seq$ so that $S_{yz}(\seq)=[yz]$.
Adding the requests to $z$ does not change $S_{xy}(\seq)$,
so $S(\seq)=[xyz]$, which implies $(x^i,z^k)\in~E$.
This shows~(\ref{transitive}).

With the help of (\ref{transitive}), we now show that if $C(x^i)=C(y^j)$,
then the pair $x^i, y^j$ is agile according to
Definition~\ref{d-agile}. We will prove this by showing that
\begin{equation}
\label{inE}
(x^i,y^j)\in E \quad\text{and}\quad (y^j,x^i)\in E.
\end{equation}
To prove (\ref{inE}), recall that $C(x^i)$ is a strongly
connected component of the graph~$G$ which also contains
$y^j$ and~$z^k$.
Therefore there exists a path in $G$ from $x^i$ to $y^j$ via~$z^k$.
This path is a sequence of unary projections $u_0,\ldots, u_n$
with $u_0=x^i$, $u_l=z^k$ for some $0<l<n$, and
$u_n=y^j$.
Let $s_i$ be the item of the corresponding unary projection
$u_i$, in particular $s_0=x$, $s_l=z$, $s_n=y$.
Ignoring the superscripts of the unary projections, we
are essentially looking at a sequence of items
$s_0s_1s_2\ldots s_n$ where $s_i\ne s_{i+1}$ for
$0\le i< n$.
We can shorten that sequence whenever $s_{q-1}$, $s_q$, and
$s_{q+1}$ are three distinct items by removing $s_q$,
because then $(s_{q-1},s_{q+1})\in E$ by (\ref{transitive}).
The problem is that we do not want to shorten it in such a
way that we cannot apply (\ref{transitive}) any more.

We call a path $u_0\ldots u_n$ between $x^i$ and $ y^j$ {\em
valid} if $|\{s_0,\ldots,s_n\}|\ge 3$.
We claim that if there exists a valid path between $x^i$ and
$y^j$ of length $n>2$, then there exists also a valid path
of length $n-1$.

To show this claim, consider the smallest $q$ so that
$s_{q-1}$, $s_q$, and $s_{q+1}$ are three distinct items.
If the path $u_0\ldots u_n$ remains valid after removing
$u_q$, we are done.
Otherwise, clearly $|\{s_0,\ldots,s_n\}|= 3$,
and removing $u_q$ makes the path invalid, which means
$s_q=z$ for some $z\not\in \{x,y\}$, and $s_q$ is the only
occurrence of $z$ in $s_0s_1s_2\ldots s_n$,
because $s_0=x$ and $s_n=y$.
We claim that $q=1$, because if $q>1$ then the sequence
$s_0s_1\ldots s_q$ is either of the form $xyxy\ldots xyz$ or
$xyx\ldots yxz$, and in both cases $s_{q-1}$ can be removed,
but $q$ was chosen smallest.
So indeed $s_1=z$, and this is the only occurrence of~$z$.
Then $s_0s_1s_2\ldots s_n$ is either of the form
$xzxy\ldots xy$ or $xzyx\ldots y$, and in each case we
can repeatedly remove $s_2$ using (\ref{transitive}), until
we arrive at $n=2$ with the sequence $xzy$.
This proves the claim.

A final application of (\ref{transitive}) then gives
$(x^i,y^j)\in E$.
The same argument shows $(y^j,x^i)\in E$.
This proves (\ref{inE}).

Because all pairs of unary projections are agile in
$C(x^i)$, we can apply Theorem~\ref{thm:operation}, whose
cases (a) and (b) prove (IVa) and (IVb).
This proves the theorem.
\end{proof}

\section{The Lower Bound for Irregular Algorithms}

In Section~\ref{sec:lowerbound} we considered deterministic
$M$-regular projective list update algorithms.
In this section, we consider randomized algorithms, which
may select deterministic algorithms that are not $M$-regular.
If this happens sufficiently rarely, the algorithm may still
be competitive.
For example, consider an algorithm that operates according
to some rule (for example \MTF), keeps track of its
incurred costs, and whenever this is a square number $Q^2$,
does not move any item for the next $Q$ requests, and then
resumes its normal operation.
This does not change its competitive ratio, but makes the
algorithm no longer $M$-regular.

In Theorem~\ref{t:Mreg} we showed that no deterministic
$M$-regular projective list update algorithm is better than
$1.6$-competitive.
For this we gave, for any $\eps>0$, a suitable distribution
on request sequences that bound the competitive ratio of the
algorithm from below by $1.6-\varepsilon$.
These request sequences are drawn from a set $\Lambda$
defined in (\ref{e:uniformdist}) with parameters $K,T$ that
are chosen sufficiently large depending on $\varepsilon$.

We extend this analysis to arbitrary randomized projective
list update algorithms using the full characterization from
the previous section.
Part of this extension involves also a sufficiently large
choice of the parameter~$\M$ in (\ref{e:uniformdist}) to
cope with algorithms that are not $\M$-regular.

In brief, the proof works as follows.
Using the crucial notion of a good state $(i,j)$
in Definition~\ref{d-good}, we call a deterministic
algorithm $\M$-regular \emph{in state} $(i,j)$
if it fulfills a certain condition, (\ref{Mreg}) below,
where the algorithm only uses the containers from
Theorem~\ref{t-containers} in the normal way that one
expects from competitive algorithms.
The lower bound from Lemma~\ref{l:good16} applies in
expectation for algorithms that fulfill condition (\ref{Mreg}).

The proof of the following theorem is mostly concerned
with the cases where the deterministic algorithm $\A$ is
``irregular'', that is, condition (\ref{Mreg}) fails.
Here we use the following argument, spelled out in detail
following (\ref{case1}):
We give simple request sequences (which depend on the
growing parameters $K,T,\M$) that have constant offline cost
but arbitrarily large cost for deterministic ``irregular''
algorithms; hence, these deterministic algorithms must be
chosen with vanishing probability.

\begin{theorem}
\label{t-irreg}
Any randomized projective list update algorithm that
accesses a list of at least three items is at best
$1.6$-competitive.
\end{theorem}

\begin{proof}
Assume the list has at least three items.
Consider a randomized projective algorithm $\calA$ and
assume that $\calA$ is $c$-competitive with $c<1.6$.
That is, there exists a constant $b$ such that
$\calA(\seq)\le c\cdot \OFF(\seq)+b$ for all request sequences $\seq$.

We adapt the proof for $M$-regular algorithms of
Section~\ref{sec:lowerbound}.
Let $\M\ge 3$, consider $\Lambda$ in (\ref{e:uniformdist}) and
consider a good state $(i,j)$ as defined in
Definition~\ref{d-good}.

Let $\A$ be a deterministic projective algorithm.
We say that algorithm $\A$ is {\em $\M$-regular in state
$(i,j)$} if, with ${\cal W}^+$ as in
Theorem~\ref{t-containers} and
$f(x^i)$ defined as in $(\ref{relcrit})$,
\begin{equation}
\label{Mreg}
C(x^i)=C(y^j)\in {\cal W}^+,
\qquad
f(x^i)< \M,
\qquad
f(y^j)< \M.
\end{equation}
It is easy to see that the proof of Lemma~\ref{l:good16}
applies if $\A$ is $\M$-regular in $(i,j)$.

Recall that $\calA$ is just a discrete probability distribution on
the set of deterministic projective algorithms.
Let $r_{ij}$ be the event that $\calA$ is $\M$-regular in
state $(i,j)$.
Analogous to (\ref{yao3}), the expected cost of $\calA$ is
bounded by considering the good states $(i,j)$ as follows:
\begin{equation}
\label{bound}
\begin{array}{rcl}
\displaystyle
E\left[\sum_{\lambda \in \Lambda}\calA(\lambda)\right]
& \ge &
\displaystyle
\sum_{(i,j)\textrm{
good}}E\left[\sum_{\lambda\in\Lambda}\calA_\lambda(i,j)\right]\\
& \ge &
\displaystyle
16KH^2T - o(KH^2T) - \sum_{(i,j)\textrm{ good}} 16(1-\prob(r_{ij})).
\\
\end{array}
\end{equation}

Let
\begin{equation}
\label{XY}
X:=HT+KH+\M
~~~~\hbox{and}~~~~
Y:=(K+1)H+\M,
\end{equation}
and recall that $H$ is a linear function of $\M$ by~(\ref{defH}).
For all good states $(i,j)$ we have, by (\ref{good}),
\begin{equation}
\label{reallygood}
  1\leq i \le X \textrm{ and } 1\leq j \le Y.
\end{equation}
If we can prove that, with growing $\M$, $K$, and $T$,
\begin{eqnarray}
\label{probrij}
  \sum_{(i,j)\textrm{ good}} 16(1-\prob(r_{ij}))
  \le  \sum_{j=1}^{Y}\sum_{i=1}^{X} 16(1-\prob(r_{ij}))
  =o(KH^2T),
\end{eqnarray}
then we have proved (\ref{e:yao2}) for irregular algorithms.

We proceed to prove (\ref{probrij}) by analyzing where
(\ref{Mreg}) fails, that is, for each of the six cases
according to
\begin{eqnarray*}
\sum_{j=1}^{Y}\sum_{i=1}^{X}(1-\prob(r_{ij}))\le
\sum_{j=1}^{Y}\sum_{i=1}^{X}\left(
\begin{array}{l}
\phantom{+}\prob(C(x^i)< C(y^j))\\
+\prob(C(x^i)> C(y^j))\\
+\prob(C(x^i)= C(y^j)\in {\cal W}^- )\\
+\prob(C(x^i)= C(y^j)\in {\cal W}_2 )\\
+\prob(C(x^i)=C(y^j)\in {\cal W}^+,~ 
f(x^i)\ge \M)\\
+\prob(C(x^i)=C(y^j)\in {\cal W}^+,~ 
f(y^j)\ge \M)
\end{array}
\right).
\end{eqnarray*}
We start by proving
\begin{equation}
\label{case1}
\sum_{j=1}^{Y}\sum_{i=1}^{X}\prob(C(x^i)< C(y^j))\le o(KH^2T).
\end{equation}
To this aim, consider the sequence $x^iy^Y$ for $1\le i \le X$.
When serving this sequence, a request to $y$ will be served
in each of the states $(i,1),\ldots, (i,j),\ldots,(i,Y)$.
Since every deterministic algorithm with $C(x^i)< C(y^j)$
pays one unit for accessing $y$ in state $(i,j)$, the
expected cost of $\calA$ for serving a request to $y$ in a
state $(i,j)$ is at least $\prob(C(x^i)< C(y^j))$.
Therefore
\begin{equation}
\label{case1b}
\calA( x^iy^Y)\ge
\sum_{j=1}^{Y} \prob(C(x^i)< C(y^j)).
\end{equation}
On the other hand,
$\calA( x^iy^Y)\le c\cdot\OFF(x^iy^Y)+b$ because
$\calA$ is $c$-competitive.
Since $\OFF(x^iy^Y)=1$  (the initial list state is $[xy]$) it follows that
\begin{equation}
\label{case1c}
\sum_{i=1}^{X}\sum_{j=1}^{Y} \prob\big(C(x^i)< C(y^j)\big)
\le \sum_{i=1}^{X} \calA( x^iy^Y)
\le X \cdot(c+b)=o(KH^2T)
\end{equation}
as desired.

The bound on $\prob(C(x^i)> C(y^j))$ is very similar, using
request sequences of the form $y^jx^X$ for $1\le j\le Y$.

For $\prob(C(x^i)= C(y^j)\in {\cal W^-} )$, we use, like for
(\ref{case1}), request sequences of the form $\seq=x^iy^Y$.
Clearly, from the first request to $y$ onwards, the critical
request to $x$ is always earlier in $\seq$ than the
critical request to~$y$.
Therefore $C(x^i)=C(y^j)\in {\cal W}^-$ implies that $y$ is behind
$x$ in the list, so
\begin{equation*}
\label{case1d}
\calA( x^iy^Y)\ge
\sum_{j=1}^{Y} \prob(C(x^i)= C(y^j)\in {\cal W^-} ),
\end{equation*}
and the same argument as after (\ref{case1b}) applies.

If $C(x^i)= C(y^j)\in {\cal W}_2$, the container
$C(x^i)$ does not contain any unary projections to items
other than $x$ or~$y$.
The list has at least a third item $z$ and either
$C(x^i)<C(z^k)$ or $C(z^k)<C(x^i)$
for any~$k$.
We consider only the first case, where we can bound
$\prob(C(x^i)<C(z^k))$ similarly to~(\ref{case1}).
By considering the request sequence $x^iz^Y$ for $1\le i\le
X$, we obtain in the same way as with (\ref{case1b}) and
(\ref{case1c}) that
$\sum_{i=1}^{X} \calA( x^iz^Y) =o(KH^2T)$.

As explained, if $C(x^i)$ and $C(y^j)$ are two containers in
${\cal W}_2$, then either $C(x^i)< C(z^k)$ or $C(x^i)>C(z^k)$
for all $z^k$ with $z\neq x,y$, so that
\[
\prob(C(x^i)= C(y^j)\in {\cal W}_2 )\le \prob(C(x^i)< C(z^k)) +\prob(C(x^i)> C(z^k)).
\]
Hence the left hand side can be bounded by the bound on the first two cases.

In a similar fashion, we bound $\prob(C(x^i)=C(y^j)\in {\cal
W^+},~f(x^i)\ge \M)$.
First of all, $C(x^i)=C(y^j)\in {\cal
W^+}$ implies that both $x^i$ and
$y^j$ are in the same container and have critical
requests, that is, $F(x^i)>0$ and $F(y^j)>0$, so the
relative requests in (\ref{relcrit}) fulfill
$f(x^i)<i$ and $f(y^j)<j$.
So $f(x^i)\ge \M$ implies $i>\M$, and therefore
\begin{equation}
\label{double}
\begin{array}{rl}
&
\displaystyle
\sum_{j=1}^{Y}\sum_{i=1}^{X}
\prob(C(x^i)=C(y^j)\in {\cal W^+},~f(x^i)\ge \M)\\
{}=&
\displaystyle
\sum_{j=1}^{Y}\sum_{i=\M+1}^{X}
\prob(C(x^i)=C(y^j)\in {\cal W^+},~f(x^i)\ge \M).\\
\end{array}
\end{equation}
Next, we show
\begin{equation}
\label{1overM}
\begin{array}{rl}
&\displaystyle
\sum_{i=\M+1}^{X}
\prob(C(x^i)=C(y^j)\in {\cal W^+},~f(x^i)\ge \M)\\
\le &
\displaystyle
\frac{1}{\M} \sum_{i'=1}^{X}
\sum_{\ell=1}^{\M}\prob(C(x^{i'+\ell})=C(y^j)\in {\cal W^+},~
  f(x^{i'+\ell})\ge \M).
\end{array}
\end{equation}
Namely, for each $i=\M+1,\ldots,X$ there are at least $\M$
choices $i',\ell$ so that $i=i'+ \ell$, be\-cause for each
$\ell=1,\ldots,\M$ the term $i' = i - \ell$ fulfills
$1\le i'\le X$.
This shows~(\ref{1overM}).

Let $1\le j\le Y$ and $1\le i'\le X$ and consider in the request
sequence $x^{i'}y^j x^\M$ (which has constant offline cost) the
last $\M$ requests to~$x$.
If their critical request is before the critical request to $y$
(which exists), they incur online cost one, so
\begin{equation}
\label{xyxM}
\begin{array} {rl}
&
\displaystyle
\sum_{\ell=1}^{\M}\prob(C(x^{i'+\ell})=C(y^j)\in {\cal W^+},~
  f(x^{i'+\ell})\ge \M)\\
\le&
\calA(x^{i'}y^j x^\M)
\le
c\cdot \OFF(x^{i'}y^j x^\M) +b =O(1).
\\
\end{array}
\end{equation}
Consider (\ref{double}), (\ref{1overM}), and (\ref{xyxM})
and note that $Y/\M=O(K)$ by (\ref{XY}) and~(\ref{defH}),
so $Y/\M\cdot X = O(KHT + K^2H)$.
This shows
\[
\sum_{j=1}^{Y}\sum_{i=1}^{X}
 \prob(C(x^i)=C(y^j)\in {\cal W^+},~f(x^i)\ge \M)
\le
Y\frac{1}{\M} X \cdot O(1) =o(KH^2T)
\]
if we let $K,T,\M$ (and thus $H$) grow while keeping $K/T$ constant.

The bound on
$\prob(C(x^i)=C(y^j)\in {\cal W^+},~f(y^j)\ge \M)$
is proved analogously to the previous bound.
\end{proof}

\section{Conclusions}

An open problem is to extend the lower bound to the full
cost model, even though this model is not very natural in
connection with projective algorithms.
This would require request sequences over arbitrarily many
items, and it is not clear whether an approach similar to
the one given here can work.

Another ambitious goal is to further improve the lower bound
in case of non-projective algorithms.
Here, the techniques of the paper do not apply at all, and
to get  improvements that are substantially larger than the
ones obtainable with the methods of \cite{AGS01} requires
new insights.

Finally, the search for good non-projective algorithms has
become an issue with our result.
Irani's SPLIT algorithm \cite{I91} is the only one known of
this kind with a competitive ratio below~2.
A major obstacle for finding such algorithms is the
difficulty of their analysis, because pairwise methods are
not applicable, and other methods (e.g.\ the potential
function method) have not been studied in depth.
We hope that our result can stimulate further research in
this direction.

A first result is a non-projective algorithm for lists of up
to four items based on partial orders which is
1.5-competitive \cite{ASW96}; for another study of
algorithms for short lists see \cite{hagerup07}.
Extending the partial order approach to longer lists is not
straightforward (and has in fact led to the lower bounds of
1.501 for lists of length five in \cite{AGS01} and
\cite[p.~38]{A02}).

{
\bibliographystyle{abbrv}
\bibliography{listupdate}
\vfill
}

\end{document}